\documentclass[num-refs]{wiley-article}

\newtheorem{secthm}{Theorem}[section]

\newtheorem{secex}[secthm]{Example}

\newtheorem{secdefn}[secthm]{Definition}
\newtheorem{secrem}[secthm]{Remark}
\newtheorem{secasm}[secthm]{Assumption}

\usepackage{comment}

\newcommand{\bE} { {\mathbb E}}
\newcommand{\bP} { {\mathbb P}}
\newcommand{\bR} { {\mathbb R}}
\newcommand{\bZ} { {\mathbb Z}}

\newcommand{\bS} { {\mathbb S}}

\newcommand{\cF} { {\mathcal F}}

\newcommand{\cS} { {\mathcal S}}
\newcommand{\cM} { {\mathcal M}}
\newcommand{\cC} { {\mathcal C}}

\newcommand{\lap} { {\operatorname{Lap}}}

\newcommand{\yk}[1]{{\color{red}#1}}

\def\red{\hfill $\lhd$}

\usepackage{siunitx}
\papertype{Original Article}
\paperfield{Journal Section}

\title{Initial State Privacy of Nonlinear Systems on Riemannian Manifolds}

\author[1]{Le Liu}
\author[2]{Yu Kawano}
\author[3]{Antai Xie}
\author[1]{Ming Cao}

\affil[1]{Faculty of Science and Engineering, University of Groningen, 9747 AG Groningen, The Netherlands}
\affil[2]{Graduate School of Advance Science and Engineering, Hiroshima University, Higashi-hiroshima 739-8527, Japan}
\affil[3]{School of Mechatronic Engineering and Automation, Shanghai University, Shanghai}

\corraddress{Le Liu, Faculty of Science and Engineering, University of Groningen, 9747 AG Groningen, The Netherlands}
\corremail{le.liu@rug.nl}

\fundinginfo{The work of Liu and Cao was supported in part by the Netherlands Organization for Scientific Research (NWO-Vici-19902). The work of Yu Kawano was supported in part by JSPS KAKENHI Grant Number JP22KK0155.}

\runningauthor{Le Liu et al.}

\begin{document}

\begin{frontmatter}
\maketitle

\begin{abstract}
In this paper, we investigate initial state privacy protection for discrete-time nonlinear closed systems. By capturing Riemannian geometric structures inherent in such privacy challenges, we refine the concept of differential privacy through the introduction of an initial state adjacency set based on Riemannian distances. A new differential privacy condition is formulated using incremental output boundedness, enabling the design of time-varying Laplacian noise to achieve specified privacy guarantees. The proposed framework extends beyond initial state protection to also cover system parameter privacy, which is demonstrated as a special application.
\end{abstract}
\end{frontmatter}

\section{Introduction}
In the modern digital era, privacy protection has emerged as a critical concern across various domains, ranging from personal data management \cite{schwartz2003property} to complex industrial systems \cite{sadeghi2015security}. As data collection technologies become increasingly pervasive, protecting privacy has become an even more important social challenge. Several privacy protection methods for static data, such as encryption \cite{kogiso2015cyber} and differential privacy \cite{mo2016privacy, kawano2021modular}, have been extended to dynamical systems. However, the majority of these methods focus on linear systems, with only a few exceptions addressing nonlinear systems, e.g., \cite{Yu2020,le2020differentially}. Furthermore, it is often assumed that data requiring protection resides in Euclidean space, even though some types of data may naturally exist on Riemannian manifolds \cite{barachant2011multiclass, arnaudon2013riemannian}. To fully utilize such geometric structures for privacy protection, there is a need for a theoretical foundation that enables privacy analysis within the framework of Riemannian geometry. To address this, in this paper, we explore a noise-adding strategy for privacy protection in a Riemannian geometric framework, focusing on initial state privacy of discrete-time nonlinear closed systems.




\subsubsection*{Literature Review}
In systems and control, privacy protection has been widely studied in the context of consensus~\cite{nozari2017differentially}, distributed optimization \cite{wang2023}, filtering \cite{le2013differentially}, tracking control \cite{Yu2020, kawano2021modular}. For privacy-preserving consensus, a decaying Laplacian noise mechanism has been employed to protect the initial state privacy \cite{nozari2017differentially}. To mitigate the degradation of consensus performance, a correlated noise-adding strategy has been explored in \cite{mo2016privacy}, which allows the preservation of the accurate consensus value while safeguarding initial state privacy. In privacy-preserving distributed optimization, a piece of private information is the local objective function of each agent. To ensure its differential privacy while still finding the optimal solution, gradient methods have been adapted accordingly \cite{wang2023}. Additionally, the alternating direction method of multipliers (ADMM) has been tailored using partially homomorphic cryptography~\cite{zhang2019}. In privacy-preserving filter design, the Kalman filter has been modified to guarantee certain privacy level for estimated state~\cite{le2013differentially}. For nonlinear systems, contraction theory has been applied to nonlinear observer design \cite{le2020differentially}. To guarantee certain privacy level of the estimated state, noise is added to the output signal. However, differential privacy analysis of the estimated state is missing because it involves analysis of nonlinear dynamics.
 For privacy preserving control, tracking controllers while preserving privacy of reference or/and initial states have been proposed in the framework of $H_\infty$-control~\cite{Yu2020,kawano2021modular} and stochastic quantized control~\cite{liu2024design}. Furthermore, initial state privacy have been studied by \cite{Yu2020, wang2023differential}. In particular, \cite{Yu2020} has dealt with nonlinear systems in terms of incremental output stability, but only in the Euclidean space. {\color{blue} On the other hand, differential privacy on Riemannian manifolds has primarily been studied in the context of static data rather than dynamical systems. For instance, Laplacian noise has been explored to ensure differential privacy on Riemannian manifolds in \cite{reimherr2021differential}. Similarly, \cite{jiang2023gaussian} has demonstrated Gaussian differential privacy on Riemannian manifolds using Gaussian noise. 
However, Riemannian manifolds remain underutilized in the privacy protection of dynamical systems.}

\subsubsection*{Contribution}
In this paper, we investigate how to protect the initial state from being inferred from a published output sequence in discrete-time nonlinear closed systems. Unlike the aforementioned literature, we employ a Riemannian distance to define the adjacency set, thereby capturing the inherent geometric structure of the initial state. Focusing on time varying Laplacian noise, we derive a sufficient condition for differential privacy in terms of incrementally output boundedness, introduced by this paper. Furthermore, we proceed with contraction analysis to provide a sufficient condition for incrementally output boundedness by utilizing a Riemannian metric. An advantage of dealing with nonlinear systems is that we can also handle protecting system parameters in addition to the initial state by treating system parameters as additional state variables. This aspect is illustrated by utilizing linear systems, noting that the problem becomes nonlinear even for such systems.


The main contributions of this paper are summarized as follows: 
\begin{enumerate} 
\item We introduce a novel definition of differential privacy utilizing the initial state adjacency set based on a Riemannian distance. 
\item We propose a sufficient condition for employing time-varying Laplacian noise to ensure a prescribed level of differential privacy, framed in terms of incremental output boundedness. 
\item We adapt the proposed method to protect system parameters in addition to the initial state. 
\end{enumerate}

\subsubsection*{Organization}
The remainder of this paper is organized as follows. In Section~\ref{sec:pre}, we formulate initial state privacy of discrete-time nonlinear systems by introducing the notion of differential privacy with respect to an adjacency set defined by a Riemannian distance. In Section~\ref{sec:main}, we provide a sufficient condition for time varying Laplacian noise to achieve a prescribed differential privacy level in terms of incremental output boundedness, which can be verified by contraction analysis. In Section~\ref{sec:app}, we apply the proposed results for protecting system parameters in addition to the initial state, which is further illustrated by numerical simulation in Section~\ref{sec:simu}.
Finally, Section~\ref{sec:con} concludes this paper.

\subsubsection*{Notation}
The sets of real numbers and integers are denoted by $\bR$ and $\bZ$, respectively. For $k_0 \in \bZ$, we define $\bZ_{k_0+}:=\bZ \cap\left[k_0, \infty\right)$ and $\bZ_{[k_0, k]}:=\bZ \cap\left[k_0, k\right]$. For some $k_0 \in \bZ$ and $k \in \bZ_{k_0+}$, subsequences of a $\mathcal{Y}$-valued sequence $(y_i)_{i \in Z}$: $\bZ \rightarrow \mathcal{Y}^{\bZ}$ are denoted by $y_{[k_0, k]} : \bZ_{[k_0, k]} \rightarrow \mathcal{Y}^{\bZ_{[k_0, k]}}$. The identity matrix of size $n \times n$ is denoted by $I_{n}$. The function $\max: \mathbb{R} \times \mathbb{R} \to \mathbb{R}$ returns the larger of two numbers. Specifically, $\max(x, y)$ is defined as $x$ if $x \geq y$, and $y$ otherwise. Given a constant $c \in \mathbb{R}$, we write $\mathbb{R}_{\geq c}$ to denote $[c, \infty) \subset \mathbb{R}$. A locally Lipschitz continuous function $\alpha: \mathbb{R}_{\geq 0} \rightarrow \mathbb{R}_{\geq 0}$ is said to belong to class $\mathcal{K}$ if it is strictly increasing and $\alpha(0)=0$. Given two functions $f: \mathcal{Z} \rightarrow \mathcal{Y}$ and $g: \mathcal{X} \rightarrow \mathcal{Z}$, the composition $f \circ g$ assigns to each $p \in \mathcal{X}$ the value $(f \circ g) (p)=f(g(p)) \in \mathcal{Y}$. The set of symmetric and positive (resp. semi) definite matrices is denoted by $\mathbb{S}_{\succ 0}^{n \times n}$ (resp. $\mathbb{S}_{\succeq 0}^{n \times n}$ ). For $P, Q \in \mathbb{R}^{n \times n}, P \succ Q$ (resp. $P \succeq Q$ ) means $P-Q \in \mathbb{S}_{\succ 0}^{n \times n}$ (resp. $P-Q \in \mathbb{S}_{\succeq 0}^{n \times n}$ ). The $\ell_1$ norm of a vector $x \in \mathbb{R}^n$ is denoted by $|{x}|$. The $\ell_2$ norm of a vector $x \in \mathbb{R}^n$ is denoted by $|{x}|_2$.

Let $\cM$ be a smooth and connected $n$ dimensional manifold equipped with a class $C^{1}$ Riemannian metric $\tilde{P} : \cM \rightarrow \mathbb{S}_{\succ 0}^{n \times n}$ and $\Gamma(x, x^{\prime})$ be the collection of class $C^{1}$ paths $\gamma:[0,1] \rightarrow \cM$ such that $\gamma(0) =x$ and $\gamma(1)= x^{\prime}$. Then, the distance function induced by a Riemannian metric $\tilde{P}(x)$ is
\begin{align}
\label{eq:dis_x}
    d_{\tilde{P}}(x, x^{\prime}) := \underset{\gamma \in \Gamma(x, x^{\prime})}{\inf}\int_{0}^{1}\sqrt{\frac{d^{\top}\gamma(s)}{ds}\tilde{P}(\gamma(s))\frac{d\gamma(s)}{ds}}ds.
\end{align}

A probability space is denoted by $(\Omega, \cF, \bP )$, where $\Omega$, $\cF$, and $\bP $ denote the sample space, the $\sigma$ algebra and the probability measure, respectively. The expectation of a random variable is denoted by $\bE[\cdot]$. $v \sim Lap(a,b)$ represents that $v$ follows an Laplacian noise with mean $a$ and diversity $b$. Specifically, its density function is given by
$$
f(v) = \frac{1}{2b}e^{-\frac{|v-a|}{b}}.
$$
 
\section{Preliminaries}\label{sec:pre}
Consider the following nonlinear autonomous discrete-time system on a smooth and connected manifold $\cM$ with a randomized output:
\begin{subequations}\label{eq:sys}
 \begin{align}
 \label{system_1}
	x_{k+1} &= f_k(x_k)\\
 \label{system_2}
    y_k &= h_k(x_k)+v_k,
 \end{align}
\end{subequations}
where $f_k:\cM \to \cM$ and $h_k:\cM \to \bR^{m}$ are of class $C^{1}$ at each $k \in \bZ$; $x_k \in \cM$, $y_k \in \bR^{m}$, and $v_k: (k, \Omega) \rightarrow \bR^m$ denote the state, output possibly eavesdropped by attackers, and noise designed for the purpose of privacy protection. 

The initial condition $x_{k_0} \in \cM$ is deterministic, where $k_0$ denotes the initial time. The solution to~\eqref{system_1} at $k \in \bZ_{k_0+}$ with the initial condition $x_{k_0}$ is denoted by $\phi_k(x_{k_0})$, i.e., $x_k = \phi_k(x_{k_0})$. Throughout the paper, following \cite{forni2013differential}, we consider a forward complete subset $\cC \subset \cM$ with respect to~\eqref{system_1}, in which any pair of points can be connected by a smooth curve $\gamma : [0,1] \rightarrow \cC$.

In this paper, our interest is to protect the initial state $x_{k_0}$ being inferred from the published output. As illustrated by Example~\ref{ex:1} below, the initial state sometimes contains a piece of private information. Moreover, considering system parameters as a part of the state, privacy-protection of system parameters can be formulated as the initial state privacy problem as illustrated later in Section~\ref{sec:app}. 

To evaluate a privacy level of the initial state, we employ the notion of differential privacy \cite{dwork2006differential} by rewriting the output subsequence $y_{[k_0, k]}$ as a function of the initial state $x_{k_0}$ and noise subsequence $v_{[k_0, k]}$ as in
\begin{align}
\label{eq:mech}
    y_{[k_0, k]} = g(x_{k_0}, k) + v_{[k_0, k]},
\end{align}
where the $i$th component of $g: \cM \times \bZ \to \bR^{k-k_0 + 1}$ is $g_i = (h \circ f_{i} \circ \cdots \circ f_{k_0})(x_{k_0})$, $k_0 \leq i \leq k$. We consider that each component of $v_k$, denoted by $v_{i,k}$, is an identical and independent distributed Laplacian noise at time step $k$, i.e.,
\begin{align}
\label{eq:noise}
    v_{i, k} \sim \lap(0, b_k), \forall i = 1, 2, \dots, m.
\end{align}
with time-dependent deterministic diversity $b_k$. The choice of Laplacian noise is based on its simplicity, but our results can be adapted to other kinds of noises such as Gaussian. Time varying Laplacian noise is commonly used in literature for privacy protection, e.g., \cite{mo2016privacy, nozari2017differentially, wang2023differential}. However, their focus are limited to the linear case in the Euclidean space.

In the reminder of this paper, we call~\eqref{eq:mech} a Laplacian mechanism, following standard notions in differential privacy \cite{dwork2006differential, le2013differentially}. Another standard notion is an adjacency relation, defined below.

\begin{secdefn}
    Given $\zeta>0$, a pair of initial states $(x_{k_0}, x'_{k_0}) \in \cC \times \cC $ is said to belong to the \emph{$\zeta$-adjacency relation} if $d_{\tilde{P}}( x_{k_0} , x'_{k_0}) \leq \zeta$, where $d_{\tilde{P}}$ is given in~\eqref{eq:dis_x}. The set of all pairs of the initial states that are $\zeta$-adjacent is denoted by $\operatorname{Adj}_d^\zeta$.
    \red
\end{secdefn}

The above adjacency relation is a generalized notion of that in literature \cite{dwork2006differential,Yu2020} because we consider a Riemannian distance. {\color{blue}When $\tilde{P} = I_n$, the Riemannian distance is nothing but the Euclidean distance. In this case, our results can be viewed as Laplacian noise counterparts of those in \cite{Yu2020} for Gaussian noise.} This generalization in Riemannian manifold is useful when handling data with inherent geometric constraints, including image data \cite{reimherr2021differential} and system information. This aspect is illustrated by the following example.

\begin{secex}
\label{ex:1}
 Consider a group of agents using an average consensus algorithm to get their meeting location. This meeting location could be the central location of their home address, where each agent's home address (initial state) is private information. Applying an average consensus algorithm gives the agent $i$'s dynamics as
 \begin{align*}
 x^{i}_{k+1} &= x^{i}_{k} - \sum_{j = 1}^{N_i} a_{i,j}(y^{j}_{k} - x^{i}_{k}),
 \quad
 x^{i}_k := \begin{bmatrix}
         x_{1, k}^{i} \\
         x_{2, k}^{i}
     \end{bmatrix}\\
 y^{i}_{k} &=  x^{i}_k + v^{i}_k
 \end{align*}
where $x_{1, k}^{i} \in \bR$ denotes the $x$-axis coordinate, representing the west-to-east direction, $x_{2, k}^{i} \in \bR$ denotes the $y$-axis coordinate, representing the south-to-west direction, $y^{i}_{k} \in \bR^2$ denotes the output of agent $i$, and $v^{i}_{k} \in \bR^2$ denotes
a multidimensional independent Laplacian noise with zero mean and diversity $b^i$.
$a_{i,j} \in (0,1)$ is the $(i, j)$ component of a stochastic matrix $A \in \bR^{N \times N}$ and $\sum_{j=1}^{N_i}a_{i,j} < 1$ for each $i$, and $N_i \in \bZ$ is the cardinality of the neighborhood set of agent $i$.

We suppose that an eavesdropper can access to all $y_{k}^j$, $j \in N_i$. Then, the initial state privacy problem of agent $i$ in average consensus is equivalent to that of a system~\eqref{eq:sys} with
 \begin{align}\label{eq:ex_para}
 x_k:= x^i_k, \quad 
 f_k(x_k) := x_{k} + \sum_{j = 1}^{N_i} a_{i,j} x_{k}, \quad
 h_k(x_k) :=  x_k.
 \end{align}
From \cite[Remark 2.7]{Yu2020} and the inequality between $\ell_1$ norm and $\ell_2$ norm, the agent~$i$ can guarantee $\epsilon$-differential privacy for $\operatorname{Adj}_d^\zeta$ with $\tilde{P} = I_2$ at any time step by choosing $b^i = \sqrt{2}\zeta/(\epsilon\sum_{j=1}^{N_i}a_{ij})$. 
  If an eavesdropper is solely focused on determining whether agent $i$ resides in the northern part of the city, the privacy of $y$-axis coordinate $x_{2,0}^i$ becomes critical. In this context, the  indistinguishable set of true information $x_{2,0}^i$ with $\epsilon$-differential privacy guarantee is $\{x_{2,0}^{i \prime} | |x_{2,0}^{i \prime} - x_{2,0}^i| \leq \sqrt{2} \zeta \}$.

  However, assume agent $i$ is living in a street and this is public information. The street can be described by an affine manifold $\cM:=\{(x_{1,0}^i, x_{2,0}^i) \in \mathbb{R}^2: x_{1,0}^i- 3 x_{2,0}^i=0\}$. Then, a Laplacian noise with the same distribution can only guarantee that the indistinguishable set of $x_{2,0}$ with $\epsilon$-differential privacy guarantee is $\{x_{2,0}^{i \prime} | |x_{2,0}^{i \prime} - x_{2,0}^i| \leq \sqrt{2}\zeta/4 \}$, since after inferring $x_{1,0}^{i}$ and $x_{2, 0}^{i}$ from $y_{k}^j$, one can further utilize the algebraic constraint $x_{1,0}^i- 3 x_{2,0}^i=0$ to narrow down the estimation of $x_{2, 0}^{i}$. This demonstrates that algebraic constraints cannot be ignored for privacy protection. 
  \red
\end{secex}

Now, we are ready to show the definition of differential privacy in terms of the adjacency relation induced by a Riemannian metric. A fundamental idea of differential privacy is to evaluate the sensitivity of a mechanism with respect to privacy-sensitive variables, which is $x_{k_0}$ in this paper. A mechanism~\eqref{eq:mech} is more private if the corresponding pair of output sequences $(y_{[k_0, k]}, y'_{[k_0, k]})$ to a pair of initial states $(x_{k_0}, x^{\prime}_{k_0})$ are close to each other. In differential privacy, the similarity of the pair of output sequences is evaluated as follows.

\begin{secdefn}
\label{def:dp}
Let $(\bR^{(k+1) m}, \cF, \bP )$ be a probability space. Given an increasing sequence $(\epsilon_k)_{k \in \bZ_{k_0+}}$, the mechanism~\eqref{eq:mech} is said to be \emph{$\epsilon_k$-differentially private} for $\operatorname{Adj}_d^\zeta$ if
\begin{align}
\label{eq:dp}
    \bP \left(g(x_{k_0}, k) + v_{[k_0, k]} \in \cS\right) \leq \mathrm{e}^{\epsilon_k} \bP 
\left(g(x^{\prime}_{k_0}, k) +v_{[k_0, k]} \in \cS\right), 
\quad \forall k\in \bZ_{k_0 +}, \; \forall \cS \in \cF
\end{align}
holds for any $(x_{k_0}, x^{\prime}_{k_0}) \in \mathrm{Adj}_d^\zeta$.
\red
\end{secdefn}

There are several remarks on Definition~\ref{def:dp}.
First, the mechanism~\eqref{eq:mech} is highly private if $\epsilon_k$ is small. Thus, it is desirable to have a smaller $\epsilon_k$ at every time instant $k$ to enhance better privacy performance. 
Second, an increasing property of $(\epsilon_k)_{k \in \bZ_{k_0+}}$ represents the fact that as more data are being collected, less private a mechanism becomes. Third, $(\epsilon_k)_{k \in \bZ_{k_0+}}$ is allowed to be an infinite sequence by requiring $\lim_{k\rightarrow \infty} \epsilon_k = \epsilon$ for some finite $\epsilon$.

\section{Differential Privacy Analysis}
\label{sec:main}
As differential privacy is defined by using a pair of output trajectories, differential privacy analysis can be reduced to incremental stability analysis as shown in \cite{Yu2020,le2020differentially}. Extending \cite[Theorem 5.1]{Yu2020} to a Laplace mechanism and a general distance function, we provide a differential privacy condition in terms of the novel concept of output incremental boundedness. Also, we show a sufficient condition for output incremental boundedness.

The following incremental stability notion plays a central role for differential analysis in this paper.
\begin{secdefn}
\label{def:ois}
The system~\eqref{eq:sys} is said to be \emph{output incrementally bounded} if there exist a class $\mathcal{K}$ function $\alpha$ and $(\lambda_k)_{k \in \bZ_{k_0+}}$ such that
    \begin{align}
    \label{eq:ois}
        |h_k(\phi_k(x_{k_0}))-h_k(\phi_k(x_{k_0}^{\prime}))| \leq \lambda_k \alpha(d_{\tilde{P}}(x_{k_0}, x_{k_0}^{\prime}))
    \end{align}
    for all $(x_{k_0}, x^{\prime}_{k_0}) \in \cC \times \cC$, $k_0 \in \bZ$, and $k \in \bZ_{k_0 +}$.
\red
\end{secdefn}

In fact, for an output incremetnally bounded system, the next theorem demonstrates how to design the Laplacian noise~\eqref{eq:noise} such that the mechanism~\eqref{eq:mech} is $\epsilon_k$-differentially private.

\begin{secthm}
\label{thm:con_privacy}
    Given $\zeta > 0$, a mechanism~\eqref{eq:mech} is $\epsilon_k$-differentially private for $\operatorname{Adj}_{d_{\tilde P}}^\zeta$ at any time instant $k \in \bZ_{k_0 +}$ for any $k_0 \in \bZ$ if the system~\eqref{eq:sys} is output incrementally bounded and time-varying Laplacian noise~\eqref{eq:noise} is designed as
    \begin{align}\label{eq:privacy}
        \frac{\lambda_k \alpha(\zeta)}{\epsilon_k - \epsilon_{k-1}} \leq b_k,
    \end{align}
    where $\epsilon_{-1} : = 0.$
\end{secthm}
{\color{blue}
\begin{proof}
     Let $\mathcal{S} = S_1 \times S_2 \times \cdots \times S_{k}$, where $S_{i} \subseteq \bR^{m}$ for all $i = 1, 2, \dots, k$. From the definition of Laplacian distribution, we have
\begin{align*} 
\mathbb{P}(y_{[k_0, k]} \in \mathcal{S} ) &=  \prod_{i=k_0}^{k}\mathbb{P}(y_i \in {S}_i)\\
&= \prod_{i=k_0}^{k} (\frac{1}{2b_i})^{m}\int_{\bR^m} 1_{S_i}(h_i(\phi_i(x_{k_0})) + v) e^{-\frac{|v|}{b_i}} d v \\
&= \prod_{i=k_0}^{k} (\frac{1}{2b_i})^{m}\int_{\bR^m} 1_{S_i}(u) e^{-\frac{\left|u-h_i(\phi_i(x_{k_0})) \right|}{b_i}} d u \\
&\leq \prod_{i=k_0}^{k} e^{\frac{\left|h_i(\phi_i(x_{k_0}))-h_i(\phi_i(x_{k_0}^{\prime}))\right|}{b_i}}  \left(\frac{1}{2 b_i}\right)^m \int_{\mathbb{R}^m} 1_S(u) e^{-\frac{\left|u-h_i(\phi_i(x_{k_0}^{\prime}))\right|}{b_i}} du \\
&\leq \prod_{i=k_0}^{k} e^{\frac{\lambda_i \alpha (\zeta
)}{ b_i}}  \left(\frac{1}{2 b_i}\right)^m \int_{\mathbb{R}^m} 1_S(u) e^{-\frac{\left|u-h_i(\phi_i(x_{k_0}^{\prime}))\right|}{b_i}} du.
\end{align*}
Since~\eqref{eq:privacy} holds, it follows that
\begin{align*}
&\prod_{i=k_0}^{k} e^{\frac{\lambda_i \alpha (\zeta
)}{ b_i}}  \left(\frac{1}{2 b_i}\right)^m \int_{\mathbb{R}^m} 1_S(u) e^{-\frac{\left|u-h_i(\phi_i(x_{k_0}^{\prime}))\right|}{b_i}} du \\
&\leq \prod_{i=k_0}^{k} e^{\epsilon_i - \epsilon_{i-1}}  \left(\frac{1}{2 b_i}\right)^m \int_{\mathbb{R}^m} 1_S(u) e^{-\frac{\left|u-h_i(\phi_i(x_{k_0}^{\prime}))\right|}{b_i}} du \\
& = e^{\epsilon_k} \mathbb{P}(y_{[k_0, k]}^{\prime} \in \mathcal{S} ).
\end{align*}
This ends the proof.
     \qed
 \end{proof}}

{\color{blue} It can be observed that our strategy ~\eqref{eq:privacy} for adding Laplacian noise is similar to the approach described in \cite[Remark 2.7]{Yu2020}. The key distinction lies in the fact that our initial state is defined on a Riemannian space, providing a broader level of generality. By tailoring the proposed approach to Gaussian mechanisms, the results in \cite{Yu2020} can also be extended to Riemannian manifolds. In this paper, we utilize a distance function induced by the Riemannian metric. However, the result remains valid even when the distance function is defined in a metric space, implying that the results in \cite{Yu2020} can also be generalized to any metric space.}

Selecting $\lambda_k$ in~\eqref{eq:ois} as $\lambda_k = \bar c \bar \lambda^{k-k_0}$ for $\bar c> 0$ and $\bar \lambda \in (0,1)$ recovers output incrementally exponentially stability \cite{yin2023output}. In this case, the system is $\epsilon_k$-differentially private with $\epsilon_k = c \sum_{i=k_0}^{k} q^{i - k_0}$ with $q \in [\bar{\lambda}, 1)$ and $c > 0$ if we select $b_k = \bar{c} \alpha(\zeta) \bar{\lambda}^{k - k_0}/(cq^{k-k_0})$. 
Furthermore, by selecting $c = \epsilon (1-q)$, we have $\epsilon_k \leq \epsilon$ for all $k \in \bZ_{k_0+}$ for arbitrary $\epsilon > 0$, i.e., the privacy budget is finite in an infinite time interval, which recovers the privacy analysis in \cite{nozari2017differentially}.

{\color{blue}At the end of this section, we provide a sufficient condition for output incrementally boundedness in case of the Riemmanian space.}

\begin{secthm}
\label{thm:oib}
A system~\eqref{eq:sys} is output incrementally bounded if there exist $c_1, c_2 > 0$ and $P(x, k) : \cM \times \bZ  \rightarrow \bS^{n \times n}_{\succeq 0}$ such that
\begin{subequations}
\label{eq:oib}
    \begin{align}
        c_1^2\frac{\partial^{\top} h_k(x)}{\partial x}\frac{\partial h_k(x)}{\partial x} \preceq P(x,k) \preceq c_2^2 \tilde{P}(x), \label{con_1}\\
        \lambda_{k}^2 \frac{\partial^{\top} f_k(x)}{\partial x} P(f_k(x),k+1) \frac{\partial f_k(x)}{\partial x} \preceq \lambda_{k+1}^2 P(x,k), \label{con_2}
    \end{align}
    \end{subequations}
    for all $(x,k) \in \cC \times \bZ_{k_0 +}$.
\end{secthm}

\begin{proof}
Let $\Gamma (x, x')$ be the collection of class $C^1$ paths $\gamma :[0, 1] \to \cC$ connecting $\gamma (0) = x$ and $\gamma (1) = x'$. Applying~\cite[Theorem 3.2]{kawano2023contraction}, we have 
\begin{align}
\label{eq:oib_hopf}
\inf_{\Gamma (x_{k_0}, x'_{k_0})} \int_0^1 \left( \frac{d^\top h(\phi_{k}(\gamma (s)))}{ds} \frac{dh(\phi_{k}(\gamma (s)))}{ds} \right)^{1/2} ds 
\le (1 + \varepsilon) \frac{c_2}{c_1} \frac{\lambda_k}{\lambda_{k_0}} d_{\tilde P} (x_{k_0}, x'_{k_0})
\end{align}
for any $(x_{k_0}, x^{\prime}_{k_0}) \times \cC \times \cC$, $k_0 \in \bZ$, $k \in \bZ_{k_0 +}$, and $\varepsilon > 0$. 

Since $h(\phi_{k}(\gamma (s)))$ is a path connecting $h(\phi_{k}(\gamma (x_{k_0})))$ and $h(\phi_{k}(\gamma (x'_{k_0})))$, the left-hand side is nothing but the squared path length of $h(\phi_{k}(\gamma (s)))$. Thus, we have 
\begin{align}
\label{eq:l_2}
|h_k(\phi_k(x_{k_0}))-h_k(\phi_k(x_{k_0}^{\prime}))|_2
= \inf_{\Gamma (x_{k_0}, x'_{k_0})} \int_0^1 \left( \frac{d^\top h(\phi_{k}(\gamma (s)))}{ds} \frac{dh(\phi_{k}(\gamma (s)))}{ds} \right)^{1/2} ds.
\end{align}
{\color{blue}Substituting \eqref{eq:l_2} into \eqref{eq:oib_hopf} yields
\begin{align*}
|h_k(\phi_k(x_{k_0}))-h_k(\phi_k(x_{k_0}^{\prime}))|_2 
\leq (1+\varepsilon ) \frac{c_2}{c_1} \frac{\lambda_k}{\lambda_{k_0}} d_{\tilde{P}}(x_{k_0}, x_{k_0}^{\prime}),
\end{align*}
where the distance function $d_{\tilde{P}}(x, x^{\prime})$ is induced by Riemannian metric $\tilde P(x)$ in \eqref{eq:oib} based on formula \eqref{eq:dis_x}. From the equivalence between vector $1$- and $2$-norms, it follows that 
\begin{align*}
|h_k(\phi_k(x_{k_0}))-h_k(\phi_k(x_{k_0}^{\prime}))|
\leq (1+\varepsilon ) \sqrt{m} \frac{c_2}{c_1} \frac{\lambda_k}{\lambda_{k_0}} d_{\tilde{P}}(x_{k_0}, x_{k_0}^{\prime}).
\end{align*}
Namely, we have \eqref{eq:ois} for $\alpha (\zeta) = (1+\varepsilon ) (\sqrt{m} c_2/c_1 k_0) \zeta$, which completes the proof.}
\qed
 \end{proof}

{\color{blue} The results can potentially be generalized to other manifolds where distance functions are defined, such as Finsler manifolds, by replacing $\tilde{P}$ with a Finsler metric. However, the extension to manifolds where distance functions are not defined remains unclear, as defining an adjacency relationship without distance is challenging.}

 \section{Application in System Parameter Protection}
\label{sec:app}
In this section, we apply the results in the previous section to protect system parameters being identified from a published output sequence. 

 Consider the following linear time-invariant system:
\begin{subequations}\label{eq:scalar}
\begin{align}
    z_{k+1} &= A(\theta) z_{k}, \\
    y_{k} &= z_k + v_{k},
\end{align}
\end{subequations}
where $\theta \in \bR$ is the private system parameter and $A(\theta) \in \bR^{n \times n}$ . By extending the system dimension through considering $\theta$ as an additional state, we obtain the following nonlinear system:
\begin{subequations}\label{eq:scalar2}
\begin{align}
    z_{k+1} &= A(\theta_k) z_{k}, \\
    \theta_{k+1} &= \theta_k,  \\
    y_{k} &= z_k + v_{k},
\end{align}
\end{subequations}
where $\theta_{k_0} = \theta$. Then, protecting system parameter $\theta$ can be formulated as the initial state privacy problem. This technique is merely standard in parameter estimation. However, even for linear systems, problems become nonlinear.

 Assuming $\theta > 0$, we consider protecting the proportion range of $\theta$. In this case, it is standard to evaluate the Rao-Fisher distance given by \cite{skovgaard1984riemannian}:
\begin{align}
\label{eq:dis_a}
    d_{\tilde{P}}(\theta, \theta^{\prime}) = |\log (\theta^{-1} \theta^{\prime})|
\end{align}
as this is widely used in medical imaging \cite{barachant2011multiclass}, radar signal processing \cite{arnaudon2013riemannian}, and continuum mechanics \cite{moakher2006averaging}. The corresponding Riemmanian metric is $\tilde P (\theta) = 1/\theta^2$.

 In fact, applying Theorems~\ref{thm:con_privacy} and~\ref{thm:oib} to~\eqref{eq:scalar2}, we have the following differential privacy condition.
     \begin{secthm}
\label{thm:sys_privacy}
Consider a system~\eqref{eq:scalar2} and suppose $|A(\theta)|_2 \leq \lambda \leq 1$ and $|dA(\theta)/d\theta|_2 \leq 1$.
    Given $\zeta > 0$, the mechanism~\eqref{eq:mech} is $\epsilon_k$-differentially private for $\operatorname{Adj}_d^\zeta$ with $\tilde P(\theta) = 1/\theta^2$ at any time instant $k$ for any $k_0 \in \bZ$ if time-varying Laplace noise~\eqref{eq:noise} is designed as
    \begin{align}
        \frac{\bar{\lambda}^{k-k_0} \zeta \sqrt{n} \max(\bar{\theta} \beta,1)
        } {\epsilon_k - \epsilon_{k-1}} \leq b_k, \quad 
        \beta := \frac{\bar{\lambda} \mu }{\bar{\lambda}^2 - \lambda^2}, \; \
        \epsilon_{k_0 -1} := 0
        \label{con_sysn}
    \end{align}
    for any $\bar{\lambda} > \lambda$ and all $(z_{k_0}, \theta_{k_0}) \in \bR \times \bR$ such that $0 < \theta_{k_0} \leq \bar{\theta}$ and $|z_{k_0}| \leq \mu$ and all $k \in \bZ_{{k_0}+}$.
\end{secthm}

\begin{proof}
     We first establish an output incrementally bounded property for system~\eqref{eq:scalar} by Theorem~\ref{thm:oib}. Let $c_1 =1$, $c_2 = \max(\bar{\theta} \beta,1) $, $\lambda_k = \bar{\lambda}^{k-k_0}$, and $P(z,k) = {\rm diag}\{I_n, \lambda^{2(k-k_0)}\beta^2\}$. Then,~\eqref{con_1} holds. Next, we confirm~\eqref{con_2}. It follows that 
    \begin{align*}
     &\lambda_k^2 \frac{\partial^{\top} f_k(x_k)}{\partial x} P(f_k(x_k),k+1) \frac{\partial f_k(x_k)}{\partial x_k} \\
     &= \bar{\lambda}^{2(k-k_0)}\begin{bmatrix}
         A^{\top}(\theta) & 0\\
         z_k^{\top}\frac{d^{\top} A(\theta)}{d\theta} & 1
     \end{bmatrix}
     \begin{bmatrix}
         I_n & 0\\
         0 & \lambda^{2(k+1-k_0)}\beta^2 
     \end{bmatrix}
     \begin{bmatrix}
         A(\theta) & \frac{d A(\theta)}{d\theta} z_k\\
         0 & 1
     \end{bmatrix} \\
     & = \bar{\lambda}^{2(k-k_0)}\begin{bmatrix}
        A^{\top}(\theta) A(\theta) & A^{\top}(\theta) \frac{d A(\theta)}{d \theta} z_k \\
         z_k^{\top} \frac{d^{\top} A(\theta)}{d \theta}A(\theta)  & z_k^{\top}\frac{d^{\top} A(\theta)}{d \theta}\frac{d A(\theta)}{d \theta}z_k + \lambda^{2(k+1-k_0)} \beta^2
     \end{bmatrix},
    \end{align*}
and consequently,
    \begin{align}\label{eq:block}
     &\lambda_{k+1}^2 P(x_k,k) - \lambda_k^2 \frac{\partial^{\top} f_k(x_k)}{\partial x_k} P(f_k(x_k),k+1) \frac{\partial f_k(x_k)}{\partial x_k} \nonumber\\ 
     &= \bar{\lambda}^{2(k-k_0)}\left( \begin{bmatrix}
        \bar{\lambda}^2 I_n - A^{\top}(\theta) A(\theta) & -A^{\top}(\theta) \frac{d A(\theta)}{d \theta} z_k \\
         -z_k^{\top} \frac{d^{\top} A(\theta)}{d \theta}A(\theta)  & \bar{\lambda}^2\lambda^{2(k+1-k_0)} \beta^2 - (z_k^{\top}\frac{d^{\top} A(\theta)}{d \theta}\frac{d A(\theta)}{d \theta}z_k +  \lambda^{2(k+1-k_0)} \beta^2)
     \end{bmatrix} \right). 
    \end{align}
Since $\bar{\lambda}^2 - A^{\top}(\theta) A(\theta) \succeq 0$, this matrix is positive semi-definite if and only if its Schur complement is positive semi-definite. 

{\color{blue}To prove this, we first estimate a lower bound on the lower right block of the matrix in the right-hand side of~\eqref{eq:block}.} Using $z_k^{\top}z_k \leq \lambda^{2(k-k_0)}\mu^2$ and $\beta = \bar{\lambda} \mu /(\bar{\lambda}^2 - \lambda^2)$, we have 
\begin{align*}
    &\bar{\lambda}^2\lambda^{2(k+1-k_0)} \beta^2 - (z_k^{\top}\frac{d^{\top} A(\theta)}{d \theta}\frac{d A(\theta)}{d \theta}z_k +  \lambda^{2(k+1-k_0)} \beta^2) \\
    &\geq \bar{\lambda}^2\lambda^{2(k-k_0)} \beta^2  - \lambda^{2(k+1-k_0)} \beta^2  - \lambda^{2(k-k_0)}\mu^2 \\
    & = \lambda^{2(k-k_0)} \mu^2 \left(\frac{\bar{\lambda}^4}{(\bar{\lambda}^2 - \lambda^2)^2} - \frac{\bar{\lambda}^2 \lambda^2}{(\bar{\lambda}^2 - \lambda^2)^2}- 1\right) \\
    & = \lambda^{2(k-k_0)} \mu^2 \frac{ \lambda^2}{\bar{\lambda}^2 - \lambda^2}.
\end{align*}
{\color{blue}Then, the Schur complement of the right-hand side of~\eqref{eq:block} satisfies}
\begin{align*}
    &  \bar{\lambda}^2 I_n - A^{\top}(\theta) A(\theta) - A^{\top}(\theta) \frac{d A(\theta)}{d \theta} z_k \left( \bar{\lambda}^2\lambda^{2(k+1-k_0)} \beta^2 - (z_k^{\top}\frac{d^{\top} A(\theta)}{d \theta}\frac{d A(\theta)}{d \theta}z_k +  \lambda^{2(k+1-k_0)} \beta^2) \right)^{-1} z_k^{\top} \frac{d^{\top} A(\theta)}{d \theta}A(\theta)\\
     & \succeq \bar{\lambda}^2 I_n - A^{\top}(\theta) A(\theta) - A^{\top}(\theta) \frac{d A(\theta)}{d \theta} z_k \left( \lambda^{2(k-k_0)} \mu^2 \frac{ \lambda^2}{\bar{\lambda}^2 - \lambda^2} \right)^{-1} z_k^{\top} \frac{d^{\top} A(\theta)}{d \theta}A(\theta) \\
     & \succeq \left( \lambda^{2(k-k_0)} \mu^2 \frac{ \lambda^2}{\bar{\lambda}^2 - \lambda^2} \right)^{-1} \left(\left( \lambda^{2(k-k_0)} \mu^2 \frac{ \lambda^2}{\bar{\lambda}^2 - \lambda^2} \right)\left(\bar{\lambda}^2 I_n - A^{\top}(\theta) A(\theta)\right) - A^{\top}(\theta) \frac{d A(\theta)}{d \theta} z_k  z_k^{\top} \frac{d^{\top} A(\theta)}{d \theta}A(\theta)\right) \\
     & \succeq \left( \lambda^{2(k-k_0)} \mu^2 \frac{ \lambda^2}{\bar{\lambda}^2 - \lambda^2} \right)^{-1} \left(\left( \lambda^{2(k-k_0)} \mu^2 \frac{ \lambda^2}{\bar{\lambda}^2 - \lambda^2} \right)\left(\bar{\lambda}^2 I_n - \lambda^2 I_n\right) - \lambda^{2(k+1-k_0)}\mu^2 I_n\right) = 0.
\end{align*}
This establishes the positive semi-definiteness by the Schur Complement Lemma.
Finally, applying Theorem~\ref{thm:con_privacy} with $\alpha =  \sqrt{n}\max(\bar{\theta} \beta, 1)$ and $\lambda_k = \bar{\lambda}^{k-k_0}$ concludes the proof. 
\qed
 \end{proof}

\section{A Numerical Example}
\label{sec:simu}
\begin{figure}
    \centering
 \includegraphics[width=\textwidth]{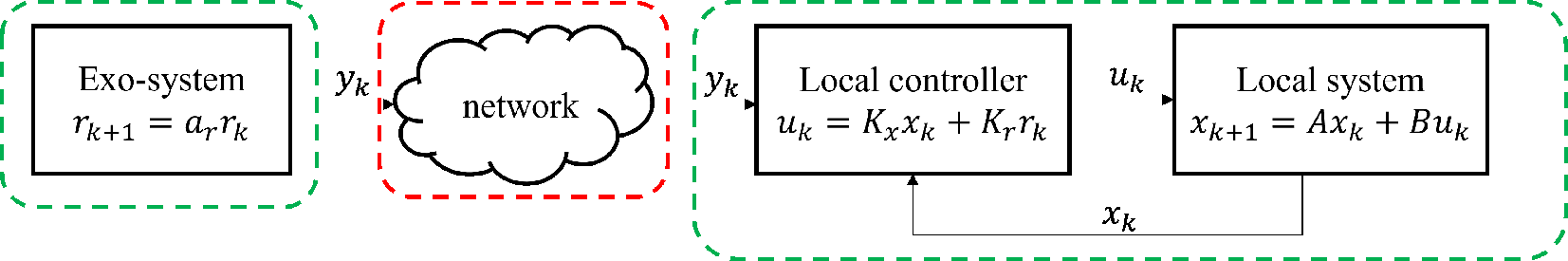}
    \caption{Mechanism diagram: The exo-system (fusion center) sends its noisy reference signals to the local system. Information in the green dashed box is not eavesdropped, while information in the red dashed box is vulnerable.}
    \label{fig:diagram}
\end{figure}
In this section, we aim to protect the privacy of the system parameter for the exo-system in an output regulation problem. The system model, including the tracking output, is given by:
\begin{subequations}\label{eq:linear}
    \begin{align}
        x_{k+1} &= A x_k + B u_k, \label{linear_1}\\
        z_k &= C x_k + D u_k, \label{linear_2}
    \end{align}
\end{subequations}
where $x \in \mathbb{R}^{n}$, $u \in \mathbb{R}^{p}$, and $z \in \mathbb{R}^m$ denote the state, control input, and tracking output, respectively. The matrices $A$, $B$, $C$, and $D$ are assumed to have appropriate dimensions. 

{The reference signal to be tracked is generated by the following exo-system:
\begin{align*}
    r_{k+1} &= A_r(\omega) r_k,\\
    y_k &= C_r r_k + v_k,
\end{align*}
where $r_k \in \mathbb{R}^{n_r}$, $y_k \in \bR^m$, and $v_k \in \bR^m$ denote the reference signal, the output, and noise generated from an independent Laplacian distribution as given in~\eqref{eq:noise} respectively, while $\omega \in \bR$ is the system parameter.
The control objective of this output regulation problem is 
\begin{align*}
\lim_{k \rightarrow \infty} e_k = 0, \quad  e_k := z_k - y_{k}. 
\end{align*}
After receiving the signal $y_k$, the local controller applies the control law
\begin{align*}
    u_k = K_x x_k + K_r y_k
\end{align*}
to regulate the system, where $K_x \in \mathbb{R}^{p \times n}$ and $K_r \in \mathbb{R}^{p \times m}$ need to satisfy the conditions in \cite[Assumption 2.4]{liu2024design}. It is shown in \cite{liu2024design} that such a pair of $K_x$ and $K_r$ can solve the output regulation problem if $v_k = 0$.

The exo-system, serving as a central coordinator, has the ability to modify amplitude or phase through different initial states. With the advancements in networking, the exo-system can transmit its output $y_k$ to the system ~\eqref{eq:linear} via a communication network, which  raises privacy concerns. The detailed system diagram with a mechanism explanation is shown in Figure~\ref{fig:diagram}. We aim to protect the privacy of system parameter $\omega$, which is considered a commercial secret and, therefore, privacy-sensitive. Relevant applications include power supply systems in semiconductor manufacturing \cite{ochs2006advanced} and current frequency regulation in electric vehicles \cite{uddin2016effects}.

{\color{blue}
The approach proposed in Section 4 provides a straightforward privacy protection strategy for $\omega$. Furthermore, this strategy can be applied to safeguard the proportional range of  $\omega$, which is particularly valuable when proportional variations reveal operational details. While most existing studies focus on the privacy protection of states or control inputs, this work addresses the protection of dynamical system parameters \cite{nozari2017differentially, wang2023,le2013differentially, zhang2019, liu2024design, wang2023differential}, which limits the availability of comparable studies. The results in \cite{Yu2020} are possibly be used to protect privacy of system parameters. However, \cite{Yu2020} focuses on privacy protection of input in the Euclidean and the results in \cite{Yu2020} cannot be directly applied to protect proportional range of $\omega$.}

Assume $A = 1$, $B =1$, $K_x =-0.3$, $C = 1$, $D= 0$, and
\begin{align*}
    A_r(\omega) &= \begin{bmatrix}
\cos (\omega) & \sin(\omega) \\
-\sin (\omega) & \cos (\omega)
\end{bmatrix}\\
    C_r &= \begin{bmatrix}
        1 & 0
    \end{bmatrix},
\end{align*}
where $\omega = \pi/20$. 
From $A_r(\omega)$ and $C_r$, one notices that $y_{k}$ is a sinusoidal wave signal with amplitude $|r_{k_0}|_2$.

Solving the equations in \cite[Assumption 2.4]{liu2024design} gives $X = [\begin{matrix}
    1 & 0
\end{matrix}]$, $U = [\begin{matrix}
    -0.0489 & 0.3090
\end{matrix}]$ and $K_r = [\begin{matrix}
    0.1511 & 0.3090
\end{matrix}]$. With these parameters, the local controller is able to track the reference signal once no noise is added. To protect the privacy of $\omega$, we design the noise $v_k$ according to Theorem~\ref{thm:sys_privacy}. It can be verified that $\|A_r(\omega)\|_2 \leq 1$ and $\|dA_r(\omega)/d\omega\|_2 \leq 1$. Additionally, we assume the initial state satisfies $\|r_{k_0}\|_2 < 300$ and $\bar{\theta} = 1$. Thus, the conditions of Theorem~\ref{thm:sys_privacy} are met. The goal of the reference system is to protect the system's information while regulating the system outputs. 

{\color{blue} We consider two differential privacy requirements with $\epsilon_k = 100 \sum_{i = k_0}^{k} 1.1^{i - k_0}$ and $\tilde{\epsilon}_k = 500 \sum_{i = k_0}^{k} 1.1^{i - k_0}$. Both of them allow for exponentially increasing privacy leakage rate. Furthermore, the adjacency parameter $\zeta$ in the adjacency set~\eqref{eq:dis_x} is defined as $\zeta = 1$. The diversity $b_k$ can be designed according to~\eqref{con_sysn} in Theorem~\ref{thm:sys_privacy} as
\begin{align*}
    b_k = \frac{\bar{\lambda}^{k-k_0} \sqrt{2} \zeta \beta}{\epsilon_k - \epsilon_{k-1}} = 22.21, \  \tilde{b}_k = \frac{\bar{\lambda}^{k-k_0} \sqrt{2} \zeta \beta}{\tilde{\epsilon}_k -\tilde{\epsilon}_{k-1}} = 4.442
\end{align*}
where $\bar{\lambda} = 1.1$ and $\beta = \frac{\bar{\lambda}  \mu}{\bar{\lambda}^2 - 1} = 1571$.} 

\begin{figure}
    \centering
    \includegraphics[width=\textwidth]{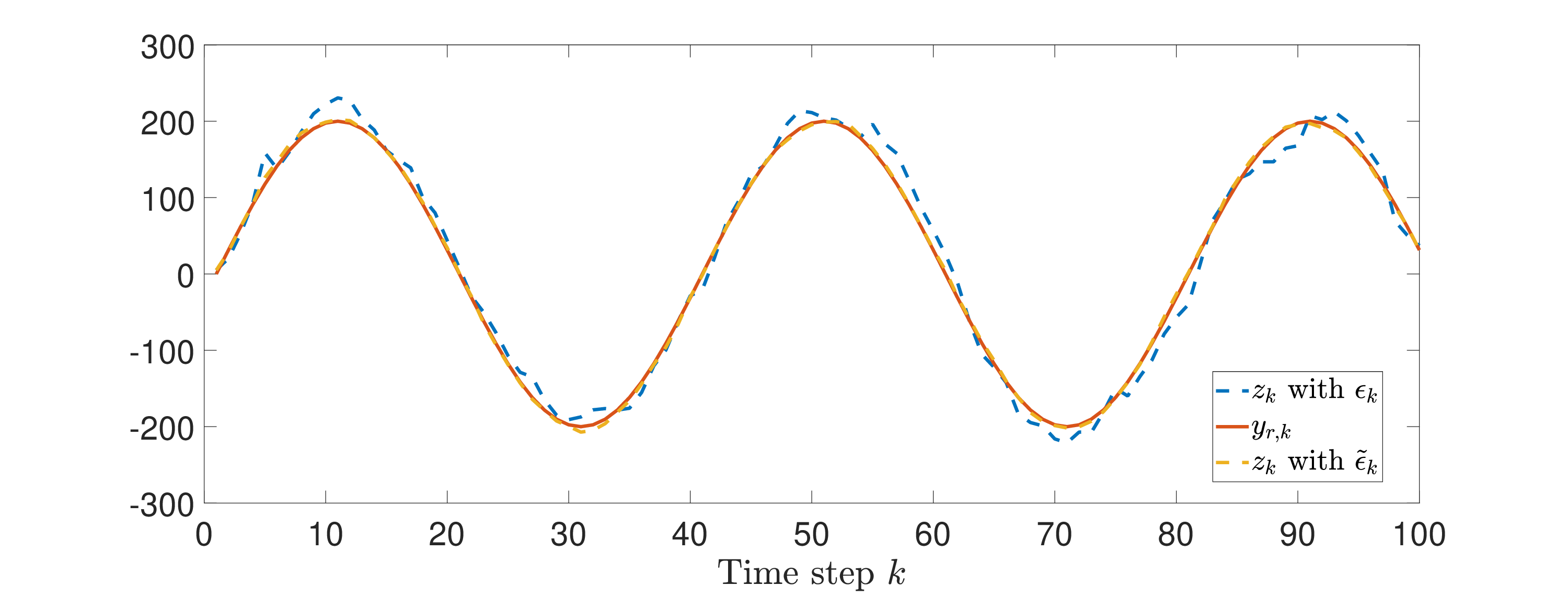}
    \caption{{\color{blue}Output regulation with system privacy protection.}}
    \label{fig:reference}
\end{figure}

\begin{figure}
    \centering
    \includegraphics[width=\textwidth]{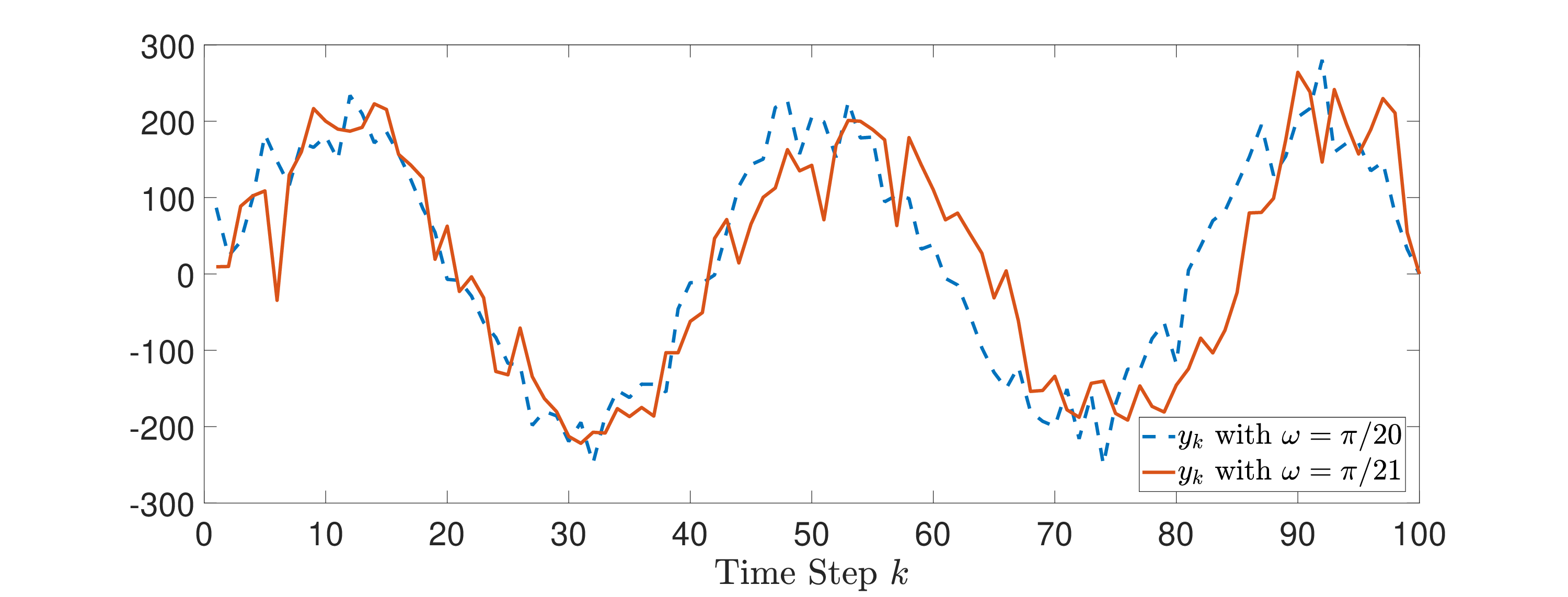}
    \caption{{\color{blue} Output Trajectories of Exo-system with Different $\omega$s}}
    \label{fig:privacy}
\end{figure}

{\color{blue}Figure~\ref{fig:reference} illustrates the output regulation process under the application of Laplacian noise. It can be observed that the system state does not precisely track the reference output during the regulation process due to the presence of noise. Since $ \tilde{\epsilon}_k < \epsilon_k$ for all $ k \in \bZ$, the privacy protection strategy using $\tilde{\epsilon}_k$ provides a higher level of privacy but results in a worse tracking performance compared to that with $\epsilon_k$. Therefore, the performance deteriorates as the privacy increases, demonstrating the trade-off between privacy and control performance}

{\color{blue}To evaluate privacy performance, Figure 3 shows output trajectories $y_k$ for two different values of $\omega$. This figure demonstrate that the outputs remain statistically similar across different $\omega$ values, indicating that privacy is effectively preserved through the proposed Laplacian mechanism.}

\section{Conclusions}
\label{sec:con}
In this paper, we have studied privacy-protection of the initial state for nonlinear closed discrete-time systems. To utilize inherent geometric structure of private information, we have introduced the concept of an initial state adjacency set based on a Riemannian distance. In accordance with this, we have proposed a differential privacy condition framed in terms of incremental output boundedness, which is verified through contraction analysis with respect to a Riemannian metric. This condition enables the design of time-varying Laplacian noise to guarantee a specified differential privacy level. Furthermore, we have demonstrated that the proposed framework can be applied to protect system parameters in addition to the initial state. Future work includes extending the proposed privacy framework to more complex privacy scenarios, involving control design.



\section*{acknowledgements}
The work of Yu Kawano was supported in part by JSPS KAKENHI Grant Number JP22KK0155. The work of Liu and Cao was supported in part by the Netherlands Organization for Scientific Research (NWO-Vici-19902).

\section*{conflict of interest}
The authors declare no conflicts of interest.

\printendnotes


\appendix

\bibliography{sample}

\end{document}